\documentclass{article}
\usepackage{graphicx}
\usepackage{color}
\usepackage{amssymb}
\usepackage[english]{babel}
\usepackage{amsmath}
\usepackage{amsthm}
\usepackage{mathtools} 
\usepackage[numbers]{natbib}
\usepackage{algorithm}
\usepackage{algorithmic}
\usepackage{auto-pst-pdf}
\usepackage{pstricks}
\usepackage{subfig}
\usepackage{caption}

\def\C{{\mathcal{C}}}
\def\D{{\mathcal{D}}}
\def\M{{\mathcal{M}}}

\theoremstyle{definition}
\newtheorem{lemma}{Lemma}
\newtheorem{corollary}{Corollary}
\newtheorem{definition}{Definition}

\newenvironment{condition}[1]{\begin{trivlist}
\item[\hskip \labelsep {\bfseries #1}]}{\end{trivlist}}
\def\ie{{\emph{i.\,e.} }}

\newtheorem*{rep@theorem}{\rep@title}
\newcommand{\newreptheorem}[2]{%
\newenvironment{rep#1}[1]{%
 \def\rep@title{#2 \ref*{##1}}%
 \begin{rep@theorem}}%
 {\end{rep@theorem}}}
\newreptheorem{theorem}{Theorem}
\newreptheorem{corollary}{Corollary}
\theoremstyle{plain}
\newtheorem{theorem}{Theorem}
 
\title{Learning-Based Constraint Satisfaction With Sensing Restrictions}

\author{Alessandro Checco, Douglas J. Leith\\Hamilton Institute, NUI Maynooth\thanks{Work supported by Science Foundation Ireland grant 11/PI/11771.}}

\begin{document}
\maketitle

\begin{abstract}
In this paper we consider graph-coloring problems, an important subset of general constraint satisfaction problems that arise in wireless resource allocation.  We constructively establish the existence of fully decentralized learning-based algorithms that are able to find a proper coloring even in the presence of strong sensing restrictions, in particular sensing asymmetry of the type encountered when hidden terminals are present.   Our main analytic contribution is to  establish sufficient conditions on the sensing behaviour to ensure that the solvers find  satisfying assignments  with probability one.   These conditions take the form of  connectivity requirements on the induced sensing graph.  These requirements are mild, and we demonstrate that they are commonly satisfied in wireless allocation tasks.  We argue that our results are of considerable practical importance in view of the prevalence of both communication and sensing restrictions in wireless resource allocation problems.   The class of algorithms analysed here requires no message-passing whatsoever between wireless devices, and we show that they continue to perform well even when devices are only able to carry out constrained sensing of the surrounding radio environment.   
\end{abstract}

\section{Introduction}

Many fundamental wireless network allocation tasks can be formulated as constraint satisfaction problems, including channel and sub-carrier allocation~\cite{DBLP:journals/corr/abs-1103-3240}, TDMA scheduling~\cite{jaume2011towards,fang2010decentralised}, scrambling code allocation~\cite{checcoself}, network coding~\cite{DBLP:journals/corr/abs-1103-3240} and so on.   Importantly, these tasks must often be solved while respecting strong communication constraints due, for example, to the range over which devices can communicate being smaller than the range over which they interfere or otherwise interact.   Recently, fully decentralised Communication-Free Learning (CFL) algorithms have been proposed for solving general constraint satisfaction problems without the need for message-passing~\cite{DBLP:journals/corr/abs-1103-3240}.    These CFL algorithms exploit local sensing to infer satisfaction/dissatisfaction of constraints, thereby avoiding the need for message-passing and use stochastic learning to converge to a satisfying assignment with probability one.   Convergence of these CFL algorithms to a satisfying assignment is, however, only guaranteed when \emph{all} devices participating in a constraint are able to sense the satisfaction/dissatisfaction of the constraint.   This sensing requirement is violated in a number of important practical problems, for example in wireless networks with hidden terminals.   The main contribution of the present paper is a new analysis of CFL-like algorithms which establishes that much weaker requirements on sensing are sufficient to guarantee convergence to a solution.   The analysis of stochastic learning algorithms is challenging, and part of the technical contribution is  the development of novel analysis tools.   We present a number of examples demonstrating the efficacy of CFL-like algorithms when subject to strong sensing as well as communication constraints, and  explore the impact of sensing constraints on the rate of convergence.

A Constraint Satisfaction Problem (CSP) consists of $N$ variables,
$\vec{x} := (x_1,\dots,x_N)$, and $M$ clauses, i.e. $\{0,1\}$-valued
functions, $(\phi_1(\vec{x}),\dots,\phi_M(\vec{x}))$.  An assignment
$\vec{x}$ is a solution if all clauses simultaneously evaluate to $1$.
In problems derived from network applications,  
each constrained variable $x_i$  is often associated to a physically distinct device, such as an access point or a base-station. 
For example,consider a collection of WLANs operating in an unlicensed
radio band.  Each WLAN can choose one of several channels to operate
on and the WLANs require to jointly select channels so as to avoid
excessive interference between the WLANs.  We can formulate this task
as a CSP by letting $x_i$ be the channel selected by WLAN
$i\in\{1,\dots,N\}$ and defining $M=N(N-1)/2$ clauses, one for each
pair of WLANs which evaluates to one if the WLANs are non-interfering,
or are out of interference range, and evaluates to zero otherwise.
Communication between the devices is impeded by multiple factors:
the interference range of a typical wireless device is
considerably larger than its communication range, and thus WLANs can
interfere but may be unable to communicate; WLANs can have different
administrative domains that would prevent communication even
via a wired backhaul and even if a proper knowledge of the physical
location of the different WLANs was known.
Consequently, the selection of the variables $x_i$ in a distributed manner (allowing message passing) is inadmissible, mandating a fully decentralized channel-selection algorithm.

A practical CSP solver for this task can only rely on each WLAN being able to measure whether or not its current choice of channel is subject to excessive interference.   Importantly, observe that this sensing need not be symmetric.  The scenario  in Figure~\ref{fig:hidden} illustrates this feature: here transmissions on link $A-B$ interfere with transmissions on link $G-H$, but not vice versa \ie transmitter $A$ acts as a hidden terminal affecting link $G-H$.  Such asymmetry in sensing is ubiquitous in networks with hidden terminals.    The analysis in~\cite{DBLP:journals/corr/abs-1103-3240} requires that all links sharing a channel are able to sense whether any one or more of the links is experiencing excessive interference and so dissatisfied, and therefore is not applicable to networks with hidden terminals.    

\begin{figure}
\centering%
\includegraphics[width=0.62\columnwidth]{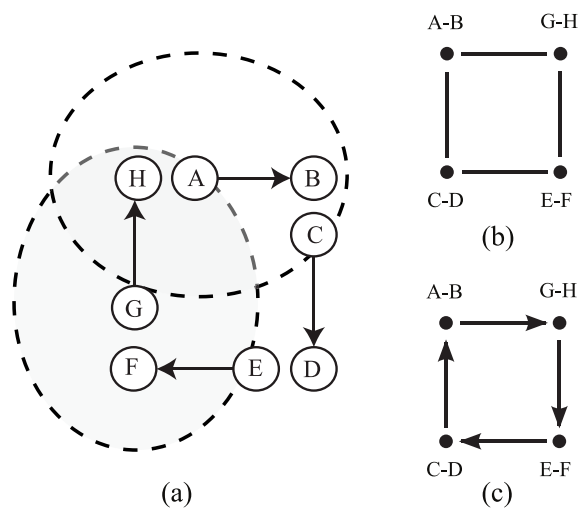}%
\caption[]{(a) Illustrating a wireless network with asymmetric sensing due to hidden terminals.   The shaded areas indicate the interference created by transmitters $A$ and $G$.  Transmissions by $A$ prevent $H$ receiving transmissions by $G$.  However, the converse is not true \emph{i.e.}\ transmissions by $G$ do not prevent $B$ from successfully receiving transmissions by $A$.    Link $A-B$ can therefore be satisfied while $G-H$ is dissatisfied.    Similarly for the other links shown.  Associating each edge with a vertex yields graph (b) corresponding to (a) for which a proper coloring is sought.  Sensing restrictions then yield the induced oriented graph (c), as explained in Section~\ref{sec:gener-decentr-cp}.}%
\label{fig:hidden}%
\end{figure}

In the present paper our aim is to address this deficiency.   We focus on graph-coloring problems, a subset of general CSPs, and constructively establish the existence of decentralized learning-based solvers that are able to find satisfying assignments even in the presence of sensing asymmetry.   We establish sufficient conditions on the sensing behaviour to ensure that the solvers find  satisfying assignments  with probability one.   We demonstrate that these conditions are commonly satisfied in wireless allocation tasks and explore the impact of sensing constraints on the speed which a satisfying assignment is found.
 
Even if in certain settings a limited amount of communication between
the devices may be possible, for example by overhearing 
traffic from some of the interferers, this information is topology
dependent and cannot be assumed during the design of the algorithm.
The opportunistic exploitation of such partial information is left for
future work.

\section{Related Work}
The graph coloring problem has been the subject of a vast literature, from
cellular networks (e.g.~\cite{raniwala2005architecture}), wireless
LANs
(e.g.~\cite{raniwala2005architecture,mishra2006distributed,mishra2006client,leung2003frequency,narayanan2002}
and references therein) and graph theory (e.g.~\cite{dousse2012percolation,kothapalli2006distributed,hedetniemi2002fault,johansson1999simple}). Almost all previous work has been concerned either with centralised schemes or with distributed schemes that employ extensive message-passing. Centralised and message-passing schemes have
many inherent advantages. In certain situations, however, these systems may not be applicable. For example, differing administrative domains may be
present in a network of WLANs.

An exception is the work of~\citet{kauffmann2007measurement,kauffmann2007self}, which proposes
a distributed simulated annealing algorithm for joint channel
selection and association control in 802.11 WLANs. However, heuristics
are used to both terminate the algorithm and to restart it if the
network topology changes. Network-wide stopping/restarting in a
distributed context can be challenging without some form of
message-passing.
 
In the field of graph theory,
\citet{dousse2012percolation,hedetniemi2002fault,johansson1999simple}
address the problem of graph coloring, when the amount of colors available is large (typically $\Delta + 1$) and allowing some form of message
passing in an undirected graph. The only exception is in~\cite{kothapalli2006distributed}, where a sort of directionality of
the graph is considered: the distributed nodes can make a choice in a
hierarchical manner, \ie when $i \rightarrow j$, node $i$ may keep its
choice even if $j$ has same color, but this is made possible assuming
the existence of this hierarchy is known and that there is still a
bidirectional channel available for communication.

This work builds upon the works of the early work of
\citet{clifford2007channel}, then refined and extended by 
\citet{leith2012wlan,DBLP:journals/corr/abs-1103-3240}; they present a
fully decentralised CFL, proven to solve a large class of problems that
include graph coloring, but without sensing restrictions.

\section{Preliminaries}

We will now introduce the problem, using similar notation
to~\cite{DBLP:journals/corr/abs-1103-3240} but extended to encompass
sensing restrictions.

\subsection{Coloring Problems (CPs)}
Let $G=(V,\M)$ denote an undirected graph with set of vertices
$V=\{1,\dots,N \}$ and set of edges $\M:=\{(i,j): i,j\in V, i\leftrightarrow j\}$, where $i\leftrightarrow j$ denotes the existence of a pair of directed edges $i\rightarrow j$, $i\leftarrow j$ joining vertices $i,j \in V$.   Note that with this notation the edges in set $\M$ are directed, since this will prove convenient later when considering oriented subgraphs of $G$.  However, since graph $G$ is undirected we have $(i,j)\in\M \iff (j,i)\in\M$.

  A coloring problem (CP) on graph $G$ with $D\in\mathbb{N}$
  colors is defined as follows.  Let $x_i \in \D$ denote the color of vertex $i$, where $\D =
  \{1,\dots,D \}$ is the set of available colors, and $\vec{x}$ denote the vector $(x_1,\dots,x_N)$.  Define clause
  $\Phi_{m} \colon \D^N \mapsto \{0,1\}$ for each edge $m = (i,j)
  \in \M$ with:
\[
\Phi_m(\vec{x})=\Phi_{m}(x_i,x_j)=\begin{cases}
1 & \text{if}\: x_{i}\neq x_{j}\\
0 & \text{otherwise}
\end{cases}.
\]
We say clause $\Phi_m(\vec{x})$ is \emph{satisfied} if $\Phi_m(\vec{x})=1$.   An assignment $\vec{x}$ is said to be satisfying if for all clauses $m\in \M$ we have $\Phi_m(\vec{x}) = 1$. That is  
\begin{equation}
  \label{eq:satisfassignment}
\vec{x} \ \text{is a satisfying assignment iff} \ \min_{m\in \M}
\Phi_m(\vec{x}) = 1.
  \end{equation}
Equivalently,  $\vec{x}$ is a satisfying assignment if and only if  $x_i \neq x_j$ for all edges $(i,j)
  \in \M$ \emph{i.e.} if $i \leftrightarrow j$.  A satisfying assignment for a coloring problem is also called a
  \emph{proper coloring}.

\begin{definition}[Chromatic Number]
The chromatic number $\chi(G)$ of graph $G$ is the smallest number of
colors such that at least one proper coloring of $G$ exists.   That
is, we require the number of colors $D$ in our palette to be greater
or equal than $\chi(G)$ for a satisfying assignment to exist.
\end{definition}

\subsection{Decentralized CP Solvers}

\begin{definition}[CP solver]
Given a CP, a CP solver realizes a sequence of vectors   $\{\vec{x}(t)\}$ such that for any CP that has a satisfying
  assignment
  \begin{description}
  \item[(D1)] for all $t$ sufficiently large $\vec{x}(t) = \vec{x}$ for
    some satisfying assignment $\vec{x}$;
  \item[(D2)] if $t'$ is the first entry in the sequence
    $\{\vec{x}(t)\}$ such that $\vec{x}(t')$ is a satisfying
    assignment, then $\vec{x}(t)=\vec{x}(t')$  for all $t>t'$.
  \end{description}
\end{definition}

In order to give criteria for classification of decentralized CP solvers, 
we re-write the LHS of Equation~\eqref{eq:satisfassignment} to
focus on the satisfaction of each variable
\begin{equation}
  \label{eq:satisfassignment2}
  \vec{x} \ \text{is a satisfying assignment iff} \ \min_{i\in V}\min_{m\in \M_i}\Phi_m(\vec{x}) = 1.
\end{equation}
where $\M_{i}$ consists of all edges in $\M$ that contain
vertex $i$, \ie
\[
\M_{i}=\left\{ (j,i)\colon(j,i)\in \M\right\}.
\]
Note that we adopt the convention of including edges in $\M_i$ which are incoming to vertex $i$, but since $(i,j)\in\M \iff (j,i)\in\M$ then $\bigcup_{i\in V}\M_i=\M$.

A decentralized CP solver is equivalent to a parallel solver, where
each variable $x_i$ runs independently an instance of the solver,
having only the information on whether all of the clauses that $x_i$
participates in are satisfied or at least one clause is unsatisfied.
The solver located at variable $x_i$ must make its decisions only
relying on this information.
\begin{definition}[Decentralized CP solver]
  A decentralized CP solver is a CP solver that for each variable
  $x_i$, must select its next value based only on the evaluation of
\begin{align}\label{eq:eval}
  \min_{m\in \M_i}\Phi_m(\vec{x}).
\end{align}
That is, the decision is made without knowing
\begin{description}
\item[(D3)] the assignment of $x_j$ for $j\neq i$.
\item[(D4)] the set of clauses that any variable, including itself,
  participates in, $\M_j$ for $j\in V$.
\item[(D5)] the clauses $\Phi_m$ for $m\in \M$. 
\end{description}
\end{definition}

\section{Coloring Problems With Sensing Restrictions}
\subsection{Decentralised Solvers}\label{sec:gener-decentr-cp}
Sensing restrictions mean, for example, that a hidden terminal is unable to sense whether or not its transmissions are causing excessive interference to the set of receivers for which it is hidden.   In other words, variable $x_i$ can only evaluate $\min_{m\in \C_i}\Phi_m(\vec{x})$ rather than $\min_{m\in \M_i}\Phi_m(\vec{x})$, where  $\C_{i}\subseteq\M_{i}$ (where equality holds only if sensing restrictions are absent).    

\begin{definition}[Decentralized CP Solver With Sensing Restrictions]
  A Decentralised CP solver where \eqref{eq:eval} is replaced with the restriction that for each variable $x_{i}$, must select its next value based only on an evaluation of
  \begin{description}
  \item[(D6)] $\min\limits _{m\in\C_{i}}\Phi_m(\vec{x})$, where
    \emph{information set} $\C_{i}\subseteq\M_{i}$ is a subset of edges incoming to node $i$ and we adopt the convention that $\min\limits _{m\in\emptyset}\Phi_{m}(\vec{x})=1$.
  \end{description}
 Note that despite the sensing restrictions we still require the solver to satisfy $(D1)$ and find a satisfying assignment, \ie for all $t$ sufficiently large $\vec{x}(t) = \vec{x}$ with $\min_{i\in V}\min_{m\in \M_i}\Phi_m(\vec{x}) = 1$.
\end{definition}

It is important to note that an assignment $\vec{x}$ may ensure $\min_{m\in\C_{i}}\Phi_m(\vec{x})=1$, $i\in V$ but might have $\min_{m\in\M_{i}}\Phi_m(\vec{x})=0$ for one or more variables and so need not be satisfying in the absence of sensing restrictions \ie it may not be a proper coloring.    We therefore require the following sensing condition in order to satisfy $(D1)$:

\begin{lemma}\label{th:D7}  Let $\C:=\cup_{i\in V} \C_i$.  Suppose that for each pair of edges $i\leftrightarrow j$ in $\M$,  at least one directed edge appears in at least one information set $C_{i}$ for some vertex $i$ \ie $(i,j)\in\M \implies (i,j)\text{ or }(j,i)\in\C$.    Then an assignment $\vec{x}$ is satisfying with sensing restrictions iff it is satisfying in the absence of sensing restrictions.  That is,  $\min_{i\in V} \min_{m\in\C_{i}}\Phi_m(\vec{x})=1$ $\iff$  $\min_{i\in V}\min_{m\in\M_{i}}\Phi_m(\vec{x})=1$.
\end{lemma}
\begin{proof}
Suppose $\min_{i\in V}\min_{m\in\C_{i}}\Phi_m(\vec{x})=1$.  That is,  $\Phi_m(\vec{x}) = 1 \ \forall m\in \C$.  By definition, $\Phi_{(i,j)}(\vec{x}) = \Phi_{(j,i)}(\vec{x})$ and since $(i,j)\in\M \implies (i,j)\text{ or }(j,i)\in\C$ the result follows.   Conversely, suppose $\min_{i\in V}\min_{m\in\M_{i}}\Phi_m(\vec{x})=1$.  Since $\C_i \subseteq \M_i$, the result immediately follows.
\end{proof}

It will be useful to consider oriented partial graph $G^\prime=(V,\C)$ induced by the information set $\{\C_1,\dots,\C_N\}$.  This graph has the same set $V$ of vertices as graph $G$ for which a proper coloring is sought, but the edges are now defined by the set of \emph{ordered} pairs $(i,j)\in \C\quad\text{if \ensuremath{(i,j)\in\C_{j}}}$.   We say $i\rightarrow j$ if there is a directed edge from $i$ to $j$, and  $i\not\rightarrow j$ if there is no edge from $i$ to $j$.   For example, Figure \ref{fig:hidden}(c) gives the graph $G^\prime$ corresponding to Figure \ref{fig:hidden}(b).  Here, the directed edge from $A-B$ to $G-H$ indicates that while $G-H$ can sense whether the edge between $A-B$  and $G-H$ is satisfied or not, $A-B$ cannot.

\subsection{Examples}
Before proceeding, we briefly demonstrate that several important resource allocation tasks in wireless networks fall within our framework of graph coloring with sensing restrictions.

\subsubsection{Channel Allocation With Hidden Terminals}
Consider a network of $N$ wireless links $i=1,\dots,N$, each consisting of a transmitter $T_i$ and a receiver $R_i$.  Let $P_i$ denote the transmit power of $T_i$ and $\gamma_{ij}$ denote the path loss between the transmitter $T_i$ of link $i$ and the receiver $R_j$ of link $j$.  The received power at $R_i$ from $T_j$ is therefore $\gamma_{ji}P_j$.   Each link can select one from a set $\D=\{1,\dots, D\}$ of available channels to use.    Link $i$ would like to select a channel in such a way that the signal power impinging on the receiver $R_i$ from other links sharing the same channel is less than a specified threshold $Q_i$ -- $Q_i$ may, for example, be selected to ensure that the SINR at $R_i$ is above a target threshold.    Each link $i$ can sense that another link $j$ is sharing the same channel when the received power $\gamma_{ji}P_j\ge Q_j$ (this might correspond to the minimum interference power that causes decoding errors on the link or to the carrier-sense threshold in 802.11).

To formulate this as a coloring problem, let  $\D=\{1,\dots,D\}$ be the palette of available colors.   Associate variable $x_i$ with wireless link $i$, $i\in \{1,\dots,N\}$, with the value of $x_i \in \D$ corresponding to the channel selected by link $i$.   Define graph $G=(V,\M)$ with $V:= \{1,\dots,N\}$ and set of edges $\M$.   Add edge $(j,i)$ to $\M$  whenever the received power $\gamma_{ji}P_j$ from link $j$ at link $i$ is above threshold $Q_i$ when both links select the same channel, $i\ne j$, $i,j\in\{1,\dots,N\}$.   Importantly, whenever an edge $(j,i)$ is in $\M$ we also add edge $(i,j)$ to $\M$, so that $G$ is an undirected graph.  A proper coloring of graph $G$  corresponds to a satisfactory channel allocation \ie $\gamma_{ji}P_j>Q_i$, for all $i\in \{1,\dots,N\}$ and all $j$ such that $x_j=x_i$ and $j\in\{1,\dots,N\}$.   

Now define graph $G^\prime=(V,\C)$ with edge $(j,i)\in\C$ when the received power  $\gamma_{ji}P_j$ from link $j$ at link $i$ is above threshold $Q_i$ when both links select the same channel.   Note that, unlike for graph $G$, we do \emph{not} also add edge $(i,j)$ to $\C$ unless $\gamma_{ij}P_i>Q_j$ when both links select the same channel.   Observe that the edges in graph $G^\prime$ embody the sensing abilities of links, and in general $\C \ne \M$ and so $G^\prime \ne G$.

Note that we can readily generalise this formulation to include, for example, the selection of multiple channels/sub-carriers by each link and to allow multiple transmitters and receivers in a link (which might then correspond to a WLAN).

\subsubsection{Decentralised TDMA Scheduling With Hidden Terminals}
When using a time division access scheme, wireless networks need to have a
schedule for accessing the channel.
This schedule can be decided in a centralized manner, but it is
possible to require a decentralized way of solving the problem.
The classical CSMA/CA approach to decentralized scheduling does not
yield convergence to a single schedule and leads to
continual collisions. Recently, there has been interest in decentralized approaches for finding collision-free schedules~\cite{fang2010decentralised}.   Consider a wireless network with $N$ links, $i=1,\dots,N$.  Time is slotted and partitioned into periodic schedules on length $T\ge N$ slots.   The transmitter on each link would like to select a slot that is different from the choice made by other transmitters if their collisions would collide (transmissions in the same slot need not collide when, for example, the two transmitters are located sufficiently far apart).   A link is able to sense whether its transmission in a slot was successful or not.

To formulate this as a coloring problem, let $\D=\{1,\dots,D\}$ be the set of available time slots in the periodic schedule.   Associate variable $x_i$ with link $i$, $i\in \{1,\dots,N\}$, with the value of $x_i \in \D$ corresponding to the slot selected by the transmitter of link $i$.   Define graph $G=(V,\M)$ with $V:= \{1,\dots,N\}$ and set of edges $\M$.   Add edge $(j,i)$ to $\M$  whenever simultaneous transmissions by the transmitters of links $i$ and $j$ would lead to failure of the transmission by $i$.   Whenever an edge $(j,i)$ is in $\M$, also add edge $(i,j)$ to $\M$.   A proper coloring of graph $G$  corresponds to a non-colliding schedule.   

Define graph $G^\prime=(V,\C)$ with edge $(j,i)\in\C$ when simultaneous transmissions by the transmitters of links $i$ and $j$ would lead to failure of the transmission by $j$.   Unlike for graph $G$, we do not also add edge $(i,j)$ to $\C$ unless simultaneous transmissions by transmitters $i$ and $j$ would lead to failure of the transmission by $j$.   Once again, the edges in graph $G^\prime$ embody the sensing abilities of links and in general $\C \ne \M$.

\section{Solving Coloring Problems With Sensing Restrictions}

\subsection{Algorithm}
Consider the stochastic learning algorithm, introduced by \citet{clifford2007channel}, described in Algorithm~\ref{alg:gCFL} with the only difference here of envisaging sensing restrictions. An instance of this algorithm is run in parallel for every variable.
\begin{algorithm}
\caption{}
\label{alg:gCFL}
   \begin{algorithmic}[1]
\STATE Initialize $p_{i,j} = {1}/{D} ,j\in \{1,\dots,D\}$.
\LOOP
\STATE Realize a random variable, selecting $x_i=j$ with probability $p_{i,j}$.
\STATE Evaluate $\min_{m\in \C_i}\Phi_m(\vec{x})$, returning
\emph{satisfied} if its value is $1$, and \emph{unsatisfied} otherwise.
\STATE Update: If \emph{satisfied},
\[
  p_{i,j} =
  \begin{cases}
    1 &\text{if $j=x_i$}\\
    0 &\text{otherwise}. 
  \end{cases}
\]
If \emph{unsatisfied},
\[
  p_{i,j} =
  \begin{cases}
    (1-b)p_{i,j} + a/(D-1+a/b) &\text{if $j=x_i$}\\
    (1-b)p_{i,j} + b/(D-1+a/b) &\text{otherwise},
  \end{cases}
\]
where $a,b \in (0,1]$ are design parameters.
\ENDLOOP
 \end{algorithmic}
\end{algorithm}

Each instance of the algorithm maintains a vector $p_{i,j}, j \in\D$, that represents the
probability of choosing color $j$ at next iteration. If satisfied, it will choose the same color with probability
one. Otherwise, the probability mass will be partially moved from color $j$ to the other colors~\cite{DBLP:journals/corr/abs-1103-3240}.
 
Algorithm~\ref{alg:gCFL} contains design parameters $a$, $b\in(0,1)$.   In the examples in this paper we select $a=1$, $b=0.1$, and do not optimise these values to particular settings.
 
In order to be a decentralised CP solver with sensing restrictions, Algorithm~\ref{alg:gCFL} must satisfy conditions $(D1)-(D6)$.   We can see immediately that Algorithm~\ref{alg:gCFL} satisfies $(D2)-(D6)$.
\begin{condition}{(D3)-(D6)}
By construction, the only information used by the algorithm
is $\min_{m \in \C_i} \Phi_m(\vec{x})$ in Step 4 and thus it satisfies the criteria
(D3), (D4), (D5) and (D6). 
\end{condition}
\begin{condition}{(D2)}
 Algorithm~\ref{alg:gCFL} also satisfies the (D2)
criterion that it sticks with a solution from the first time one
is found. To see this, note that the effect of Step 5 is that if a
variable experiences success in all clauses $\Phi$ that it participates in
it continues to select the same value with probability 1. Thus
if all variables are simultaneously satisfied in all clauses, \ie
if $\min_{m \in \C_i} \Phi_m(\vec{x})$, then the same assignment will
be reselected indefinitely with probability 1.  
\end{condition}
It remains to verify satisfaction of $(D1)$, \ie convergence of the algorithm to a satisfying assignment, which is the subject of the next section.

\subsection{Convergence Analysis}

Recall the following definition:
\begin{definition}[Strongly Connected Graph]
  A path of length $q$ in oriented graph $G^\prime=(V,\C)$ is a sequence \(\mu = (u_1, u_2,\dots, u_q)\) of edges in
  $\C$ such that  the terminal endpoint of edge $u_i$ is the initial endpoint of
  edge $u_{i+1}$ for all $i < q$.  Oriented graph $G^\prime=(V,\C)$ is strongly connected if it contains a path starting in $x$ and ending in $y$, for each pair of distinct vertices $x\neq y\in V$.
\end{definition}
We now state our main analytic result:
  \begin{theorem}
    \label{thm:bounds}
  Consider any satisfiable coloring problem with graph $G=(V,\M)$ and information sets   $\{\C_1,\dots,\C_N\}$.  Suppose:
      \begin{description}
    \item[(A)] At least one half of each undirected edge $i\leftrightarrow j$ in $\M$ appears in
at least one information set $C_{i}$ for some vertex $i$, \ie $(i,j)\in\M \implies (i,j)\text{ or }(j,i)\in\C$;
    \item[(B)] The induced graph $G^\prime=(V,\C)$ is strongly connected.
    \end{description}
  Then  with probability greater than $1-\epsilon \in (0,1)$, the
  number of iterations for Algorithm~\ref{alg:gCFL} to find a satisfying
  assignment is less than
  \[
  (N^3)\exp(N^4\log(\gamma^{-1}))\log(\epsilon^{-1})\
  \text{where } \gamma = \frac{\min(a,b)}{D-1+a/b}.
  \]
  \end{theorem}
 \begin{proof}
 See Appendix.
 \end{proof}  
 As Theorem~\ref{thm:bounds} covers any arbitrary CP that admits a
 solution, for any given instance these bounds are likely to be
 loose. They do, however, allow us to conclude the following corollary
 proving that if a solution exists, Algorithm~\ref{alg:gCFL} will
 almost surely find it:
\begin{corollary}\label{thm:issolver}
For any coloring problem that admits a proper coloring and that
fulfills conditions (A) and (B), Algorithm~\ref{alg:gCFL} will find a proper coloring in
almost surely finite time. 
\end{corollary}

Intuitively, we expect that sensing restrictions may increase the time it takes to find a satisfying assignment.   When $\C_i=\M_i$, $i\in V$ (perfect sensing) then $\C=\M$ and our analysis yields the following bound on the convergence rate:
\begin{corollary}
\label{thm:coroll}
When $\C_i = \M_i \forall i\in V$, with probability greater than $1-\epsilon \in (0,1)$, the
  number of iterations for Algorithm~\ref{alg:gCFL} to find a satisfying
  assignment is less than 
  \[
  (N)\exp(\frac{N(N+1)}{2}\log(\gamma^{-1}))\log(\epsilon^{-1}).
  \]
\end{corollary}
 \begin{proof}
 See Appendix.
 \end{proof}  
 That is, our upper bound on convergence rate is improved from $N^4$
 to $N^2$ with perfect sensing. This corresponds to the bound found
 in~\cite{DBLP:journals/corr/abs-1103-3240} for generic DCS problems,
 but it is looser than the refined bound found there for graph
 coloring problems.  However, it is important to stress that this
 observation comes with the caveat that, as already noted, we believe
 both of these bounds are extremely loose.  Hence, we revisit this
 question below using numerical simulations, which yield tight
 measurements of convergence rate.

\subsection{Relaxing Strong Connectivity Requirement}
 \label{sec:relax-strong-conn}

The requirement in Theorem~\ref{thm:bounds} for the sensing graph $G^\prime$ to be strongly connected can be relaxed in a number of ways.  If graph $G$ is not connected, we only have to ask for strong connectivity separately for the induced graph corresponding to each connected component.   More generally, we can extend our analysis to situations where graph $G$ consists of a number of strongly connected components with sufficiently sparse interconnections between these components. 

\begin{figure}
\centering%
\includegraphics[width=0.65\columnwidth]{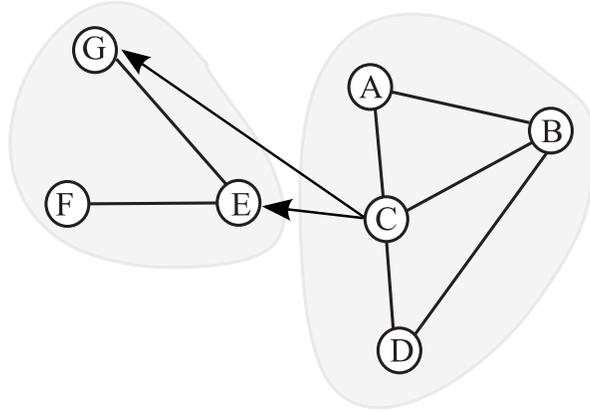}%
\caption{Example of a graph $G$ with two strongly connected components ($\{A,B,C,D\}$ and $\{E,F,G\}$) which are sparsely interconnected.  The chromatic number $\chi(G)$ of graph $G$ is $4$, the chromatic numbers of the connected components are $3$ and $2$ respectively.}\label{fig:ex2}%
\end{figure}

To help gain insight, consider the example graph $G$ shown in Figure~\ref{fig:ex2}.   Graph $G$ consists of two strongly connected components, $\{A,B,C,D\}$ and $\{E,F,G\}$, with two directed edges between them.   Subgraph $\{A,B,C,D\}$ has no incoming edges and can be colored on its own (\ie without reference to the rest of graph $G$).  Component $\{E,F,G\}$ has two incoming edges.  Observe that these can be thought of as, in the worst case, reducing by two the number of colors available in our palette $\D$ when coloring $\{E,F,G\}$.  Now, $\D$ contains at least $\chi(G)=4$ colors (since we assume coloring of graph $G$ is feasible) while subgraph $\{E,F,G\}$ is colorable using only two colors.  Hence, regardless of the colors of vertices $C$ and $D$ on the two incoming edges, sufficient colors are always available to color subgraph $\{E,F,G\}$.  We formalize these observations in the following theorem.

\begin{definition}[Subgraph of $G^\prime$ generated by $V_k$]
The subgraph of graph $G^\prime=(V,\C)$ generated by $V_k$ is the
 graph $(V_k,\{(i,j): i,j\in V_k, (i,j)\in \C\})$.  That is, the graph
with $V_k$ as its vertex set and with all the arcs in $G^\prime$ that have both their
endpoints in $V_k$. With a slight abuse of notation, we will identify
the subgraph with the vertex set $V_k$ that generates it. 
\end{definition}

\begin{definition}[In-degree of a subgraph]
  The in-degree of the subgraph  graph $G^\prime=(V,\C)$ generated by $V_k$, denoted by
  $deg(V_k)$, is the number of vertices $j \in V\setminus V_k$ that have at
  least one edge $(i,j) \in \C$, $j\in V_k$ ending in $V_k$.
\end{definition}

\begin{theorem}\label{thm:two}
Let $V = \bigcup_{k=1}^p V_k, \ V_i\bigcap V_j = \emptyset$ be a partition of the vertex set $V$ of oriented graph $G^\prime=(V,\C)$ such that (i) the subgraph generated by $V_k$, $k\in\{1,\dots,p\}$ is strongly connected and (ii) the subgraph generated by the union $\cup_{k\in S} V_k$ of any subset $S\subset\{1,\dots,p\}$ is not strongly connected.   That is, directed edges may exist between strongly connected components, but their union is not strongly connected.   Let $D$ be the number of colors available in our palette $\D$ and let $\chi(V_k)$ be the chromatic number of the (undirected) subgraph of $G=(V,\M)$ generated by $V_k$.   Suppose that  
\begin{equation}
  \label{eq:subgraphs}
  \chi(V_k)  \le D-deg(V_k),\,k=1,\dots,p
\end{equation}
Then for any coloring problem that admits a proper coloring and that
fulfills condition (B) of Theorem~\ref{thm:issolver}, Algorithm~\ref{alg:gCFL} will find a proper coloring in
almost surely finite time. 
\end{theorem}
  \begin{proof}
The main idea is that if a strongly connected component $V_k$ requires
less colors than $D$ to be colored, and if the number of edges
entering in $V_k$ is small enough, as shown in Equation~\eqref{eq:subgraphs} and in Figure~\ref{fig:ex2}, then $V_k$ can be colored by Algorithm~\ref{alg:gCFL} even if some vertices $j \not\in V_k$ are not reachable by any $i \in V_k$, with $i\leftarrow j$.
The original coloring problem is satisfiable by hypothesis, so we have at least $\chi(G)$ available colors $D$ in our palette.  We need to consider two cases. Case 1: $deg(V_k)=0$.  Since $\chi(V_k) \le \chi(G)$ (since $V_k$ is a subgraph of $G$), at least one proper coloring of subgraph $V_k$ exists and we can use Theorem~\ref{thm:issolver} to establish that  Algorithm~\ref{alg:gCFL}  will almost surely find a proper coloring.  Case 2: $deg(V_k)>0$.   The incoming edges reduce by at most $deg(V_k)$ the choice of the colors available for subgraph $V_k$.  Hence, provided $\chi(V_k) \le D-deg(V_k)$ then we can apply Theorem~\ref{thm:issolver} to subgraph $V_k$ in isolation from the rest of graph $G$ to establish that Algorithm~\ref{alg:gCFL}  will almost surely find a proper coloring.
\end{proof}

 \section{Performance on Random Graphs}
The upper bound in Theorem~\ref{thm:bounds} is a worst case bound, and in addition we believe that it may not be tight.  Hence, it is important to also evaluate the performance of Algorithm~\ref{alg:gCFL} using numerical measurements.   In this section we present measurements for a class of random graphs that are based on an idealised model of wireless network interference.  These graphs have been widely studied~\cite{dousse2012percolation} and provide a method for technology-neutral evaluation.   In Section~\ref{sec:manhattan-data} we evaluate performance in a technology specific manner using graphs derived from the WiGLE database of 802.11 hot spot locations.

\subsection{Random Graph Model}
\label{sec:descr-rand-model}

We use realizations drawn from the Directed Boolean Model (DBM) described
in~\cite{dousse2012percolation}.  The vertices of the graph are drawn
from a Poisson point process in $[0,1]^2$ with intensity $\lambda$
(with appropriate re-scaling to a required area -- in the examples here we rescale to an area of $100\,m^2$).
In the original undirected Boolean model (also known as the blob
model~\cite[see][Section 10.5]{grimmett1999percolation}), each vertex
is the center of a closed ball of random radius. The radii of the
balls are independently and identically distributed. The (undirected)
connectivity graph is obtained by adding an edge between all pairs of
points whose balls overlap, \ie $B(y) \cap B(z) \neq \emptyset$, where
$B(y), B(z)$ denote the balls centered on vertices $y, z$
respectively.  To obtain a directed graph, following
\cite{dousse2012percolation} we slightly change the above rule and put
a directed edge between $y$ and $z$ if $z \in B(y)$ and an edge
between $z$ and $y$ if $y \in B(z)$. This modified model is referred to
as the Directed Boolean Model (DBM).

In our measurements the radii are chosen uniformly at random from
the finite set of the coverage areas corresponding to transmitting powers in the range
$12-20$\,dBm, with steps of $2$\,dBm, and a specified detection threshold $R$.  We use the 3GPP path loss
model for indoor environments~\cite{ue2004}, based on the Okumura-Hata log-distance model
\[
PL_{dB}(d) = 43.3 \cdot \log_{10}d + 11.5 + 20 \cdot \log_{10} f
\]
where $d$ is the distance in meters and $f$ is the frequency in
GHz. In the examples here we select fixed frequency $f=2.412$\,GHz.  For detection
threshold $R$ in dB and transmit power $P$ in dB, the coverage radius
is then given by
\[
d: PL_{dB}(d) + P \ge R
\]
Figures~\ref{fig:ex2} and \ref{fig:ex3} show examples of graph generated using this model.

We focus in the most challenging cases by selecting the number $D$ of available colors equal to the minimum feasible value $\chi(G)$.

\subsection{Meeting Connectivity Requirements}
\begin{figure}[t]
\centering%
\includegraphics[width=0.8\columnwidth]{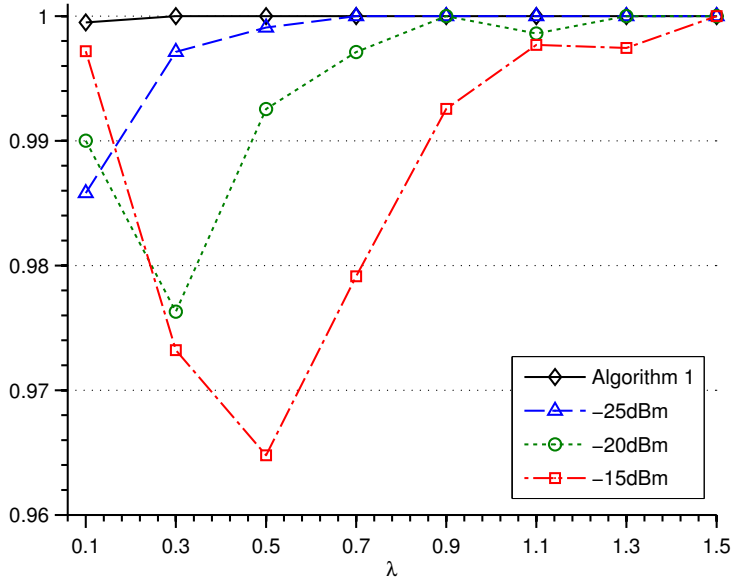}%
\caption{Fraction of DBM graphs nodes satisfying connectivity
  requirements of Theorem~\ref{thm:two} versus the detection
  threshold.  Additionally, the fraction of nodes correctly colored by Algorithm~\ref{alg:gCFL} for detection threshold of $-25$\,dBm is shown.}\label{fig:frac}%
\end{figure}
Theorems~\ref{thm:bounds} and \ref{thm:two} place connectivity requirements on the induced sensing graph $G^\prime$ in order to ensure that Algorithm~\ref{alg:gCFL} converges to a satisfying assignment.   We begin by evaluating the fraction of random graphs in the Directed Boolean Model that meet these requirements.   Figure~\ref{fig:frac} plots this fraction for a range of detection thresholds $R$ and vertex densities $\lambda$.    It can be seen that for detection thresholds below $-15$\,dBm greater than 96\% of graphs satisfy the connectivity requirements.   Figure~\ref{fig:ex3} shows some examples of some DBM graphs corresponding to a $-25$\,dBm threshold.  Observe that they consist of a number of connected components and so the relaxed connectivity conditions provided by Theorems~\ref{thm:two} are of considerable importance here.   Note also that modern wireless devices typically have a noise floor of less than $-70$\,dBm and so $-25$\,dBm is conservative.  

Moreover, Figure~\ref{fig:frac} shows the measured fraction of
vertices for which Algorithm~\ref{alg:gCFL} successfully found a
satisfying assignment for a detection threshold of $-25$\,dBm.  It can
be seen that greater than 99.9\% of vertices are successfully colored
by the algorithm. For $\lambda=0.5$ and detection threshold of
$-15$\,dBm, 0.04\% of the vertices that does not fulfill the  conditions of
Theorem~\ref{thm:two} are still correctly colored by
Algorithm~\ref{alg:gCFL}. This small gap can be explained with the
fact that the conditions of Theorem~\ref{thm:two} are sufficient, but
not necessary for convergence: some topologies can lead to convergence
for their particular structure or because of a fortunate initial
condition (see Figure~\ref{fig:ex3} for some examples).
\begin{figure}
\centering%
\includegraphics[width=0.8\columnwidth ]{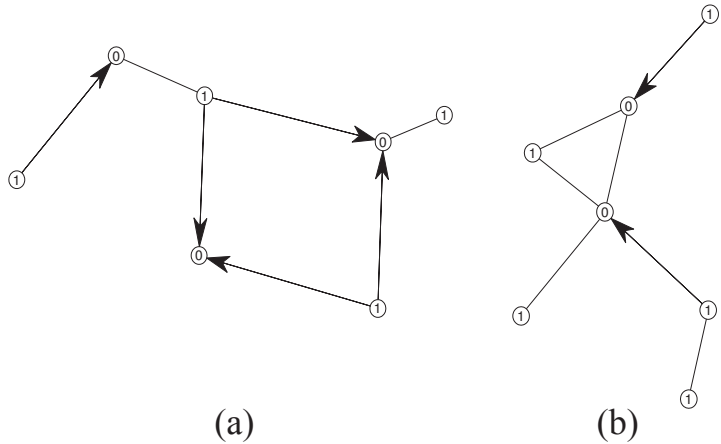}%
\caption{Example DBM graphs. The nodes labeled with $1$ are the one
  that satisfy the connectivity conditions of Theorem~\ref{thm:two} }\label{fig:ex3}
\end{figure}

\subsection{Convergence Rate}
\begin{figure}
\centering%
\includegraphics[width=0.9\columnwidth]{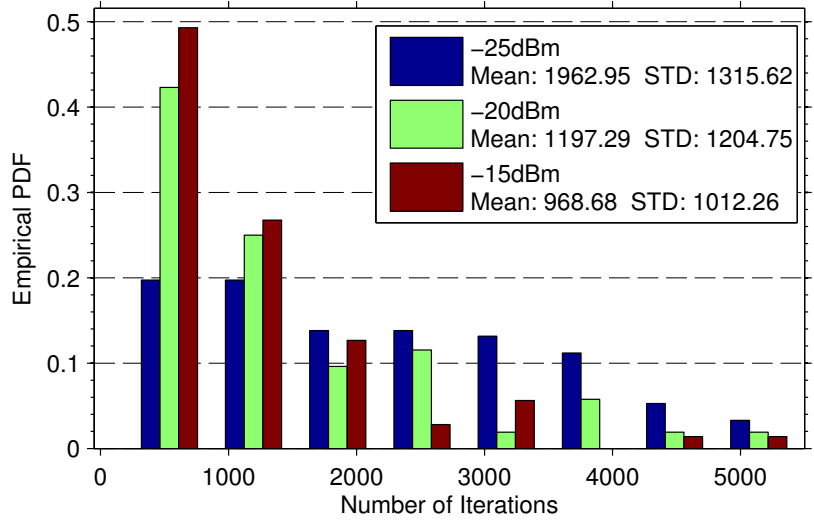}%
\caption{Measured convergence rate of Algorithm~\ref{alg:gCFL} for
  DBM graphs using a number of available colors equal
  to the chromatic number $\chi$ of the graph for three different
  detection thresholds. The density is $\lambda=0.5$.}\label{fig:cflblob}
\end{figure}
Figure~\ref{fig:cflblob} shows the measured distribution of convergence time for Algorithm~\ref{alg:gCFL} versus the detection threshold used for sensing.   For a threshold of $-25$\,dBm, the mean convergence time is less than 2000 iterations.   When the required threshold is increased to $-15$\,dBm, the mean convergence time decreases to less than 1000 iterations.    These measurements are for a link density of $\lambda=0.5$, corresponding to on average 50 wireless links in an area of $100m^2$.  Recall that we selected the number $D$ of available colors equal to the minimum feasible $\chi(G)$, thereby focussing on the most challenging situations.   For larger numbers of colors it can be verified that the convergence time decreases exponentially in the number of colors above $\chi(G)$.

The comparison of the bounds given by Theorem~\ref{thm:bounds} with
the case without sensing restrictions given by
Corollary~\ref{thm:coroll} suggests that sensing restrictions lead to
an increase in the convergence time. This is indeed the case, as shown
in Figure~\ref{fig:cflvs}, where the convergence rate of
Algorithm~\ref{alg:gCFL} is shown with and without sensing
restrictions for DBM graphs with $\lambda = 0.5$ and detection
threshold of $-15$\,dBm. However for DMB graphs it can be seen that this
increase is small.
\begin{figure}
\centering%
\includegraphics[width=0.9\columnwidth]{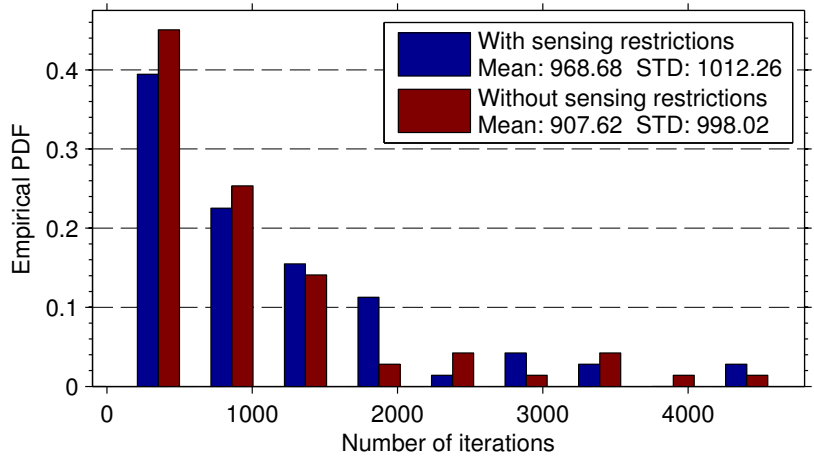}%
\caption{Measured convergence rate of Algorithm~\ref{alg:gCFL} for
  DBM graphs with detection threshold of $-15$\,dBm and density of
  $\lambda = 0.5$, with and without sensing restrictions.}\label{fig:cflvs}
\end{figure}

We also analyzed in Section~\ref{sec:convergence-time} the impact of the number of available colors on the
convergence time.

\section{Case Study: Manhattan WiFi Hots Spots}
\label{sec:manhattan-data}
\begin{figure}
\centering%
\includegraphics[width=\columnwidth]{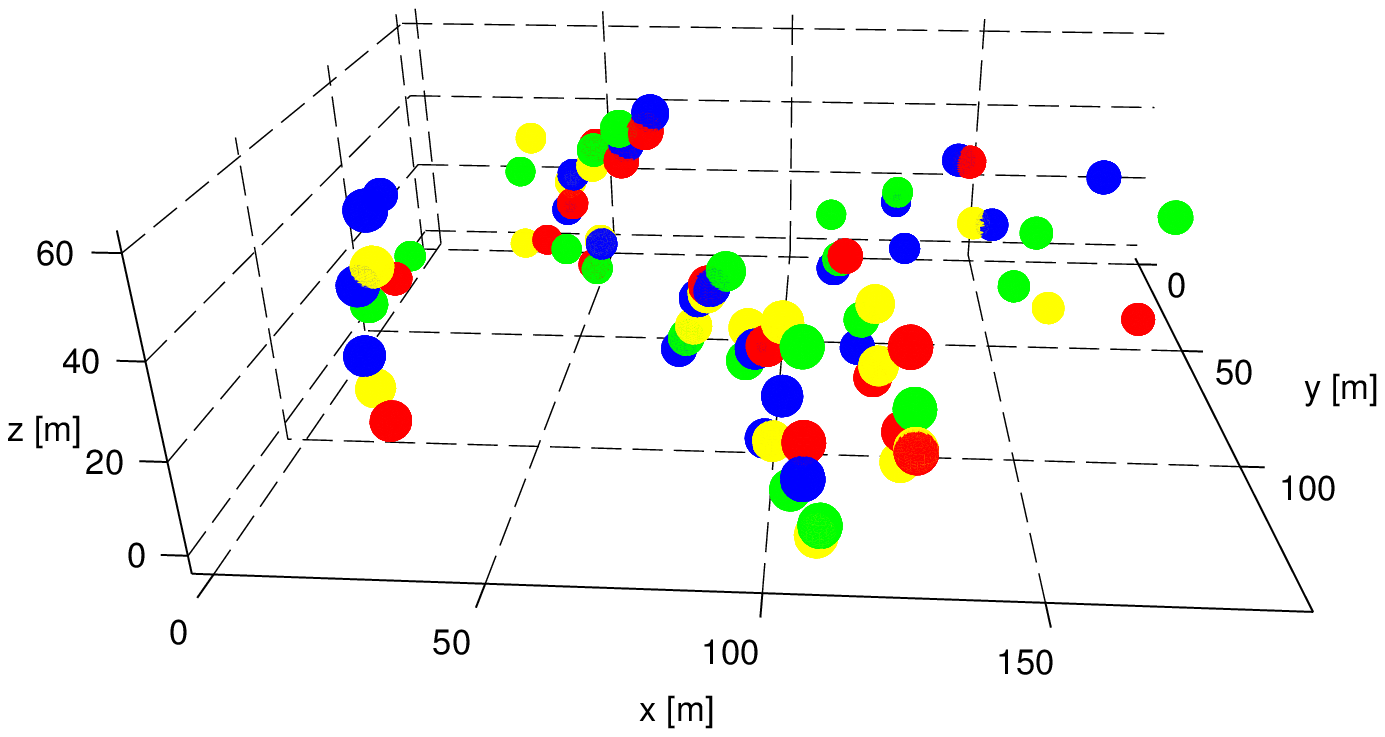}%
\caption{Example assignment for Manhattan WiFi hot spots.   Wireless access points are indicated by points and the color indicates the radio channel selected by the AP.}\label{fig:ex4}
\end{figure}
From the online database WiGLE~\cite{wigle} we obtained the locations of WiFi wireless Access Points
(APs) in an approximately $150\,m^2$ area at the junction of 5th Avenue
and 59th Street in Manhattan\footnote{The extracted (x,y,z) coordinate
  data used is available online at www.hamilton.ie/net/xyz.txt}.    This space contains 81 APs utilizing the IEEE 802.11 wireless standard.   We model radio path loss with distance as $d^\alpha$, where $d$ is the distance in meters and $\alpha=4.3$ is the path loss exponent (consistent with the 3GPP indoor propagation model~\cite{ue2004}), and the AP transmit powers are selected uniformly at random in the range $12-20$\,dBm, with steps of $2$\,dBm.  The aim of each AP is to select its radio channel in such a way as to ensure that it is sufficiently different from nearby WLANs.   This can be written as a coloring problem with $N=81$ APs and $N$ variables $x_i$ corresponding to the channel of AP $i$, $i=1,\dots,N$.   As per the 802.11 standard~\cite{80211} and FCC regulations, each AP can select from one of 11 radio channels in the 2.4\,GHz band and so the $x_i$, $i = 1,2,\dots,N$ take values in $\D = \{1,2,\dots,11\}$.   To avoid excessive interference each AP requires that the received signal strength from other APs sharing the same channel is attenuated by at least $-60$\,dB.  When all APs use the maximum transmit power of $18$\,dBm allowed by the 802.11 standard, this requirement is met when the received power is less than $-45$\,dBm and ensures that the SINR is greater than $20$\,dB (sufficient to sustain a data rate of $54$\,Mbps when the connection is line of sight and channel noise is Gaussian~\cite{goldsmith2005wireless}).

The APs do not belong to a single administrative domain and so a decentralised solver is required.  The presence of hidden terminals means that the solver must find a satisfying solution while subject to sensing asymmetry.

\begin{figure}
\centering%
\includegraphics[width=0.8\columnwidth]{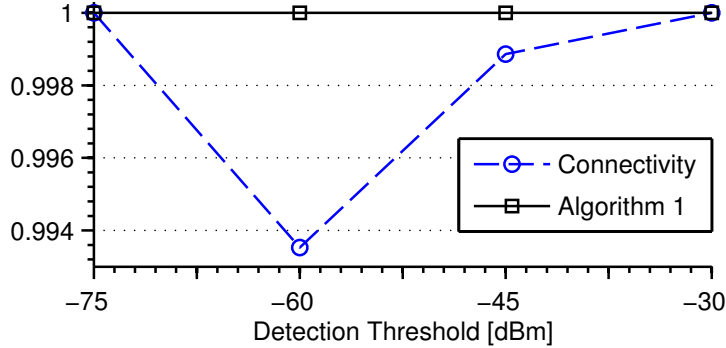}%
\caption{Connectivity for Manhattan WiFi hot spots.   Each measure is
  the fraction of nodes that satisfy connectivity requirements of
  Theorem~\ref{thm:two}. The fraction of nodes correctly colored by
  Algorithm~\ref{alg:gCFL} is also shown.}\label{fig:connect}
\end{figure}
\begin{figure}
\centering%
\includegraphics[width=0.8\columnwidth]{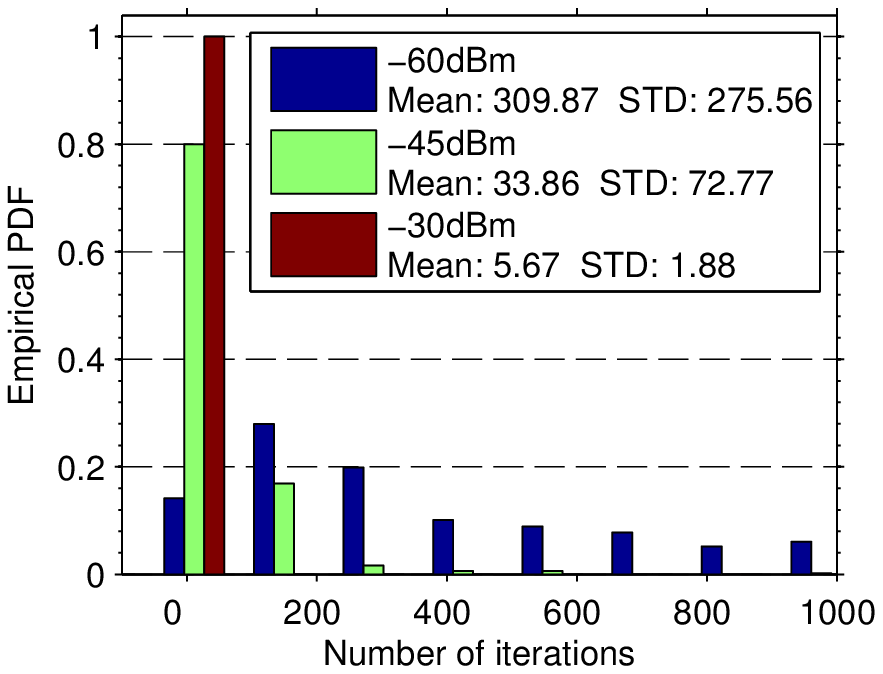}%
\caption{Measured convergence rate of Algorithm~\ref{alg:gCFL} for
  Manhattan WiFi hot spots using a number of available colors equal
  to the chromatic number $\chi$ of the graph.}
\label{fig:rate0}%
\end{figure}

\begin{figure}
\centering%
\includegraphics[width=0.8\columnwidth]{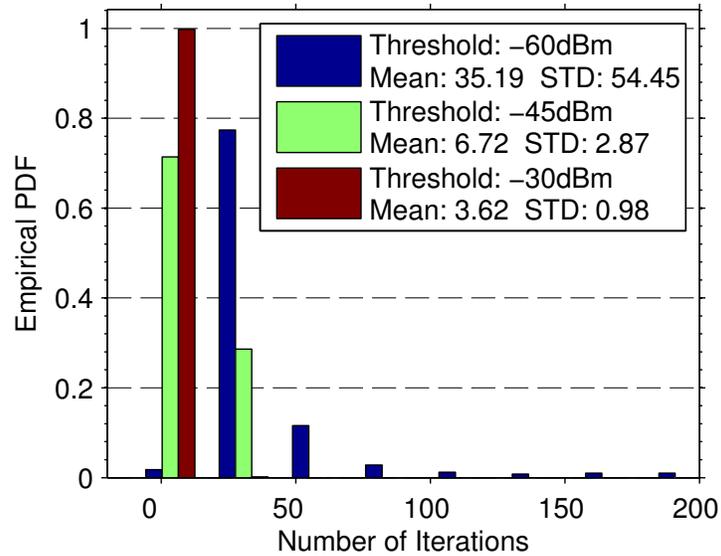}%
\caption{Measured convergence rate of Algorithm~\ref{alg:gCFL} for
  Manhattan WiFi hot spots using a number of available colors equal
  to $\chi +2$, where $\chi$ is the chromatic number of the graph.}
\label{fig:rate2}%
\end{figure}
 
The connectivity requirement of Theorem~\ref{thm:two} was observed to
be satisfied $>99\%$ of examples, see Figure~\ref{fig:connect}.
\subsection{Convergence Time}
\label{sec:convergence-time}
Algorithm~\ref{alg:gCFL} was observed to converge in less than 1000
iterations in all examples.  Figure~\ref{fig:rate0} shows the measured distribution of convergence
time for Algorithm~\ref{alg:gCFL} versus the detection threshold used
for sensing.  For a threshold of $-45$\,dBm, corresponding to the target
requirement noted above, the mean convergence time is less than 34
iterations.  In a prototype lab set-up we have shown that an update
interval of less than 10 seconds is feasible on current 802.11
hardware. Thus the mean time to convergence is under 6 minutes, which
is a reasonable time-frame for practical purposes.  When the required
threshold is increased to $-30$\,dBm, the mean convergence time
decreases to less than 6 iterations \ie under 1 minute when each
iteration takes 10 seconds.

To examine the impact of the number of available colors we compare, in
Figures~\ref{fig:rate0} and~\ref{fig:rate2}, the convergence time when the number of
colors is equal to $\chi$ and $\chi+2$ respectively.
For a detection threshold of $-60$\,dBm, adding two colors reduces the
mean convergence time of almost 10 times.
\section{Conclusions}
In this paper we focus on graph-coloring problems, a subset of general
CSPs.  We constructively establish the existence of decentralized
learning-based solvers that are able to find satisfying assignments
even in the presence of sensing restrictions, in particular sensing
asymmetry of the type encountered when hidden terminals are present.
Our main analytic contribution is to  establish sufficient conditions
on the sensing behaviour to ensure that the solvers find  satisfying
assignments  with probability one.   These conditions take the form of
connectivity requirements on the induced sensing graph.  These
requirements are mild, and we demonstrate that they are commonly
satisfied in wireless allocation tasks.  We explore the impact of
sensing constraints on the speed which a satisfying assignment is
found, showing the increase in convergence time is not significant in
common scenarios.

Our results are of considerable practical importance in view of the
prevalence of both communication and sensing restrictions in wireless
resource allocation problems.  The class of algorithms analysed here
requires no message-passing whatsoever between wireless devices, and
we show that they continue to perform well even when devices are only
able to carry out constrained sensing of the surrounding radio
environment.

Future work includes the extension of our analysis to general decentralised constraint satisfaction problems and more refined results for specific classes of graphs.

\section*{Appendix: Proofs}
We will exhibit a lower bound for the probability of a sequence of
events that ultimately lead to an increase in the number of properly
colored vertices. Such a sequence can be quite complicated in cases where a node $i$ is
unsatisfied by a node $j$ such that $i \not \rightarrow j$ (asymmetric sensing), because in this case it is necessary to propagate the
dissatisfaction to $j$ via another path, and do so in a way that allows us
to restore the original color of the other vertices.

Consider graph $G^\prime=(V,\C)$.  Let
\[
A=\bigcup\left\{ \vec{x}\colon\Phi_m(\vec{x})=1\text{ for all }m\in\C\right\},
\]
denote the set of assignments which are absorbing for Algorithm~\ref{alg:gCFL} and 
\[
  B=\bigcup\{\vec{x} \colon x_i\neq x_j \ \text{for all}\
  i\leftrightarrow j \},
\] 
the set of proper colorings, with $A \supseteq B$.   Under condition $(A)$ of Theorem~\ref{thm:bounds}, $A=B$ and all absorbing assignments are also satisfying.   When the coloring problem is feasible then $A\neq \emptyset$ (at least one satisfying assignment exists).  Let $a \in A$ be a target satisfying assignment.  We will refer to the assignment at time step $t$ as $\vec{x}(t)$. 
Let $F_{\vec{x}(t)}$ denote the set of vertices that have their
target color, \ie $F_{\vec{x}(t)} = \{ i  \colon i\in V, x_i(t)=a_i
\}$.  Furthermore, let $U_{\vec{x}(t)}$ denote the set of unsatisfied
vertices, \ie $U_{\vec{x}(t)} = \{ i  \colon i,j\in V, x_j(t)=x_i(t),
j\rightarrow i \}$, where  $i \rightarrow j$ and $j
 \leftarrow i$ denote the existence of an oriented edge $(i,j) \in \C$.
Define $\gamma = \min(a,b)/(D-1 + a/b)$.
\begin{lemma}
  \label{lem:obvious}
 If a vertex is unsatisfied, when using Algorithm~\ref{alg:gCFL}  the probability that the vertex  chooses any color $j$
 at the next step is greater than or equal to $\gamma$.
 \begin{proof}
   This follows from step 5 of Algorithm~\ref{alg:gCFL}.
 \end{proof}
\end{lemma}

\begin{lemma}
\label{lem:1}
Given any satisfiable CP and an information set $\{\C_1,\dots,\C_N\}$
with starting unsatisfied assignment $\vec{x}(0) \in \D^N, \vec{x}(0)
\not\in A$ such that $F_{\vec{x}({0})} \not\supseteq
U_{\vec{x}({0})}$, Algorithm~\ref{alg:gCFL}  will reach an assignment
$\vec{x}({\tilde{t}})$ such that $F_{\vec{x}({\tilde{t}})} \supsetneq
F_{\vec{x}(0)}$ and $F_{\vec{x}({\tilde{t}})} \supseteq
U_{\vec{x}({\tilde{t}})}$ in $\tilde{t}\le|F_{\vec{x}({\tilde{t}})}|-
|F_{\vec{x}(1)}|\le N$ steps with probability greater than
\[\gamma^{\sum\limits_{k=|F_{\vec{x}(0)}|}^{|F_{\vec{x}({\tilde{t}})}|}\!\!\!\!\!\!\!\!
  k} < \gamma^{N(N+1)/2}\] In other words,  all vertices that had their target color
in $\vec{x}(0)$ will still have it in $\vec{x}({\tilde{t}})$, and all
unsatisfied vertices in $\vec{x}({\tilde{t}})$ will have their target
color.
  \begin{proof}
   At the first step we consider the event that changes the assignment
   to 
    \begin{equation}
      \label{eq:lem1step}
      x_i(1) =
      \begin{cases}
        a_i & \text{if } i \in U_{\vec{x}(0)}, \\
        x_i(0) & \text{otherwise}.
      \end{cases}
    \end{equation}
    This event is feasible since Algorithm~\ref{alg:gCFL}  ensures that all
    satisfied vertices will remain unchanged and each unsatisfied
    vertex may change its color.
    The probability that this event happens is greater than
    $\gamma^{|U_{\vec{x}(0)}|}$.  After this step we have \(
    F_{\vec{x}(1)} = F_{\vec{x}(0)} \cup U_{\vec{x}(0)}.  \) Now, the
    set of unsatisfied variables $U_{\vec{x}(1)}$ could have
    changed. If $U_{\vec{x}(1)} \subseteq F_{\vec{x}(1)}$, we have finished,
    otherwise we consider again the event that changes the assignment
    similarly to     equation~\eqref{eq:lem1step}, \ie at generic step $t$ we have
    \begin{equation*}
      x_i(t) =
      \begin{cases}
        a_i & \text{if } i \in U_{\vec{x}(t-1)}, \\
        x_i({t-1}) & \text{otherwise}.
      \end{cases}
    \end{equation*}
    The probability of this happening is greater than
    $\gamma^{|U_{\vec{x}(t-1)}|}$, and it can be lower bounded by
    $\gamma^{|F_{\vec{x}(t)}|}$ because $F_{\vec{x}(t)} =
    F_{\vec{x}(t-1)} \cup U_{\vec{x}(t-1)}$. Since while $U_{\vec{x}({{t-1}})} \not\subseteq
    F_{\vec{x}({{t-1}})}$ we have $F_{\vec{x}({t})}$
    is a strictly growing set, and we have a finite number of vertices
    $N$, a finite time $\tilde{t}\le N$ exists  after which
    we will necessarily have $U_{\vec{x}({\tilde{t}})} \subseteq
    F_{\vec{x}({\tilde{t}})}$.
    The worst case in regards to the number of steps is when at each
    step, only one new vertex is added to $F_{x(t)}$, giving us the
    bound for the number of steps of $|F_{x({{\tilde{t}}})}| -
    |F_{x(1)}|$.
  \end{proof}
\end{lemma}

\begin{lemma}
  \label{lem:cycle}
  Consider any satisfiable CP and an information set
  $\{\C_1,\dots,\C_N\}$ with induced graph $G^\prime=(V,\C)$ and
  color $x_i(t)\in \D$ associated with each vertex $i\in V$ at time
  $t$.  Let $A\subset \D^{|V|}$ denote the set of satisfying
  assignments.  Suppose $|V|>1$, $\vec{x}(0)\not\in A$ (the initial
  choice of colors is not a satisfying assignment) and graph
  $G^\prime$ is strongly connected.  Let $a\in A$ be an arbitrary
  satisfying assignment.  If $F_{\vec{x}(0)}\supseteq U_{\vec{x}(0)}$,
  there exists a pair of vertices $k,j$ with same color such that $j \rightarrow k$ and
  $k \not \rightarrow j$, with $k \in U_{\vec{x}(0)}$ and $j \not\in
  U_{\vec{x}(0)}, j \not \in F_{\vec{x}(0)}$; in other words, a vertex
  $k$ exists that is unsatisfied by a satisfied vertex $j$ that
  doesn't have final color, and the minimum path length between $k$ and $j$ is greater
  than $1$.

  \begin{proof}
    Consider any unsatisfied vertex $k \in U_{\vec{x}(0)}$. At least
    one such vertex   exists
    because $\vec{x}(0)\not\in A$. The hypothesis $F_{\vec{x}(0)}
    \supseteq U_{\vec{x}(0)}$ ensures $x_k(0)=a_k$.  Since $k$ is
    unsatisfied, there exists a node $j$ such that $j \rightarrow k$ and
    $x_j(0) = x_k(0)$, and $x_j(0)\neq a_j$, because $x_k(0)=a_k$,
    $x_j(0)=x_k(0)=a_k$ and since $k\leftrightarrow j$ in $G$, we must
    have $a_j\neq a_k$. The hypothesis $F_{\vec{x}(0)} \supseteq
    U_{\vec{x}(0)}$ also ensures that $j$ is satisfied, because if it
    was unsatisfied it should have its final color because
    $F_{\vec{x}(0)} \supseteq U_{\vec{x}(0)}$ and this would
    contradict the property just proved that $x_j(0)=a_j$.  Since $j$
    is satisfied and it has same color than $k$, we have $k \not
    \rightarrow j$.
  \end{proof}
\end{lemma}

\begin{definition}[1-rotation]
  A \emph{1-rotation} is an operator $P$ acting on vector
  $\vec{s}=(s_1,s_2,\dots,s_{m})$, $m>1$, such that $P(\vec{s})_i =
  s_{i+1}$, $i=\{1,2,\dots,m-1\}$ and $P(\vec{s})_{m} = s_{1}$.
  Repeating a 1-rotation $m$ times yields the identity operation, \ie
  $P^{m}(\vec{s}) = \vec{s}$.
\end{definition}

\begin{lemma}
\label{lem:perm}
Consider any satisfiable CP and an information set $\{\C_1,\dots,\C_N\}$
and induced graph $G^\prime=(V,\C)$ and color $x_i(t)\in \D$
associated with each vertex $i\in V$ at time $t$.  Suppose there
exists a cycle $p_1\rightarrow p_2 \rightarrow \dots\rightarrow p_{m}
\rightarrow p_{m+1} \rightarrow p_1 \subseteq G'$, $m>1$, with
$x_{p_{m+1}}(0) = x_{p_{1}}(0)$ at time $t=0$.  Let
$\vec{s}(0)=(x_{p_1}(0),x_{p_2}(0),\dots,x_{p_{m}}(0))$. With
probability greater than $\gamma^{Nm}$, after $m$ time steps Algorithm~\ref{alg:gCFL}  will realize a \mbox{1-rotation} of the vector $\vec{s}(0)$,
\ie $\vec{s}(m) =
P(\vec{s}(0))=(x_{p_2}(0),x_{p_3}(0),\dots,x_{p_{m}}(0),x_{p_1}(0))$,
while leaving the colors of all other vertices unchanged.
\begin{proof}
  Observe that at time $t=0$ vertex $p_1$ is unsatisfied since
  $x_{p_{m+1}}(0) = x_{p_1}(0)$ and $p_{m+1}\rightarrow p_1$.
  Consider the event that at time $t=1$
 \begin{align*}
    x_{p_1}(1) &= x_{p_2}(0) \\
    x_{p_2}(1) &= x_{p_2}(0) \\
     & \ \, \vdots \\
    x_{p_m}(1) &= x_{p_m}(0)
  \end{align*}
  and the colors of all other vertices remain unchanged.  This event
  is feasible since Algorithm~\ref{alg:gCFL}  ensures that all satisfied
  vertices will remain unchanged and each unsatisfied vertex may
  choose any color from set $D$ with probability at least $\gamma$.
  From the latter, the event described occurs with probability greater
  than $\gamma^N$.  Observing that vertex $p_2$ is now unsatisfied
  since $x_{p_{1}}(1) = x_{p_2}(0)= x_{p_2}(1)$ and $p_{1}\rightarrow
  p_2$, suppose that at time $t=2$
\begin{align*}
    x_{p_1}(2) &= x_{p_1}(1) = x_{p_2}(0) \\
    x_{p_2}(2) &= x_{p_3}(1) = x_{p_3}(0) \\
    x_{p_3}(2) &= x_{p_3}(1) \\
     & \ \, \vdots \\
    x_{p_m}(2) &= x_{p_m}(1).
  \end{align*}
  Again this event is feasible and occurs with probability greater
  than $\gamma^N$.  After $m$ such steps we have $\vec{s}(m) =
  P(\vec{s}(0))$ as claimed, and this sequence of events will occur
  with probability greater than $\gamma^{Nm}$.
\end{proof}
\end{lemma}
\begin{lemma}
\label{lem:2}
Consider any satisfiable CP and an information set
$\{\C_1,\dots,\C_N\}$ with induced graph $G^\prime=(V,\C)$ and
color $x_i(t)\in \D$ associated with each vertex $i\in V$ at time $t$.
Let $A\subset \D^{|V|}$ denote the set of satisfying assignment.
Suppose $\vec{x}(0)\not\in A$ (the initial choice of colors is not a
satisfying assignment) and graph $G^\prime$ is strongly connected. Let
$d \in \D$ be an arbitrary color.  Let $k\in V$ be an unsatisfied
vertex and let $j$ be a vertex such that $j\rightarrow k$,
$x_j(0)=x_k(0)$ (at least one such vertex exists since $k$ is
unsatisfied).  With probability greater than $\gamma^{N^3}$, in
$\tilde{t} < N^2$ steps Algorithm~\ref{alg:gCFL}  will choose
$\vec{x}(\tilde{t})$, such that $x_i(\tilde{t})=x_i(0) \ \forall
i\in\{i:i\in V, i\neq j\}$ and $x_j(\tilde{t})=d$.
  \begin{proof}
    Since $G'$ is strongly connected, there exists a cycle
    $k\rightarrow\dots\rightarrow j \rightarrow k \subseteq G'$.  Let
    us relabel the $m+1>1$ vertices in the cycle using the ordering
    induced by the cycle, \ie $p_1= k, p_{m+1} = j$ and so
    $p_1\rightarrow p_2\rightarrow \dots\rightarrow p_{m+1}
    \rightarrow p_1$. Define vector $\vec{s}(t)=(x_{p_1}(t),
    \dots,x_{p_{m-1}}(t),x_{p_{m}}(t))$.  We need to consider two
    cases.
    $m=1$.  In this case the cycle is $p_1\rightarrow p_2
      \rightarrow p_1$.  By assumption, $x_{p_{2}}(0)=x_{p_1}(0)$ and
      so vertex $p_2$ is unsatisfied since $p_1\rightarrow p_2$.  It
      follows that, with probability at least $\gamma^N$, after 1 time
      step Algorithm~\ref{alg:gCFL}  will realize the event that vertex $p_2$
      selects color $d$ and the color of all other vertices remains
      unchanged.
    $m>1$.  Using Lemma~\ref{lem:perm}, with probability greater
      than $\gamma^{Nm}$ in $m$ steps Algorithm  \ref{alg:gCFL}  will realize
      a 1-rotation of the vector $\vec{s}(0)$ \ie
      $\vec{s}(m)=(x_{p_2}(0), \dots,x_{p_{m}}(0),x_{p_{1}}(0))$ leaving
      the colors of all other vertices unchanged.  Observe that vertex
      $j = p_{m+1}$ must now be unsatisfied because
      $x_{p_{m}}(m)=x_{p_1}(0)$, $x_{p_{m+1}}(m)=x_{p_{m+1}}(0)=x_{p_1}(0)$
      and $p_{m}\rightarrow p_{m+1}$.  Now consider the event at time
      $m+1$ where vertex $p_{m+1}$ takes the color of vertex $p_1$
      (and the color of all other vertices remains unchanged).  This
      event occurs with probability greater than $\gamma^N$.  After
      $m+1$ steps we have $\vec{s}(m+1)=(x_{p_2}(0),
      \dots,x_{p_{m}}(0),x_{p_{1}}(0))$ and
      $x_{p_{m+1}}(m+1)=x_{p_2}(0)$, and this event occurs with
      probability greater than $\gamma^{N(m+1)}$.
      Applying again Lemma~\ref{lem:perm}, after a 1-rotation and
      changing the color of unsatisfied vertex $p_{m+1}$ we have
      $\vec{s}(2m+2)=(x_{p_3}(0), \dots,x_{p_{1}}(0),x_{p_{2}}(0))$
      and $x_{p_{m+1}}(2m+2)=x_{p_3}(0)$.  This state is reached after
      $2(m+1)$ steps with probability greater than $\gamma^{2N(m+1)}$.
      Repeating, after $m(m+1)$ steps $\vec{s}(m^2+m)=(x_{p_1}(0),
      \dots,x_{p_{m-1}}(0),x_{p_{m}}(0))$ and $x_{p_{m+1}}(m^2+m)=d$
      (where at the very last step we select the color of unsatisfied
      vertex $p_{m+1}$ to equal $d$ rather than the color of $p_1$).
      This state is reached after $m(m+1)$ steps with probability
      greater than $\gamma^{Nm(m+1)}$.  Since $m\le N$, $m(m-1) < N^2$
      steps and $\gamma^{N m(m-1)} > \gamma^{N^3}$.
\end{proof}
\end{lemma}
\begin{lemma}
\label{lem:3}
Consider any satisfiable CP and an information set
$\{\C_1,\dots,\C_N\}$ with induced graph $G^\prime=(V,\C)$ and
color $x_i(t)\in \D$ associated with each vertex $i\in V$ at time $t$.
Let $A\subset \D^{|V|}$ denote the set of satisfying assignments.
Suppose $|V|>1$, $\vec{x}(0)\not\in A$ (the initial choice of colors
is not a satisfying assignment) and graph $G^\prime$ is strongly
connected.  Let $a\in A$ be an arbitrary satisfying assignment.  If
$F_{\vec{x}(0)}\supseteq U_{\vec{x}(0)}$  with probability
greater than $\gamma^{N^3}$, in $\tilde{t} \le N^2$ steps Algorithm~\ref{alg:gCFL}  will reach an assignment $\vec{x}(\tilde{t})$ such that
$F_{\vec{x}(\tilde{t})}\supsetneq F_{\vec{x}(0)}$ and
$|F_{\vec{x}(\tilde{t})}| = |F_{\vec{x}(0)}| + 1$;
\begin{proof}
  Lemma~\ref{lem:cycle} ensures a pair $i, j$ exists such that $x_k(0)
  = a_k, x_k(0) \in U_{\vec{x}(0)}$ and $x_j(0) = x_k(0)$.
  Lemma~\ref{lem:2} ensures that, in less than
  $N^2$ steps, with probability greater than $\gamma^{N^3}$, Algorithm~\ref{alg:gCFL} will reach an assignment in which vertex $j$ assumes color
  $a_j$ and the colors of all other vertices are unchanged.
\end{proof}
\end{lemma}
\begin{proof}[Proof of Theorem~\ref{thm:bounds}]
  Consider Algorithm~\ref{alg:gCFL}  starting from an assignment $\vec{x}(0)$. Select an
  arbitrary valid solution $a\in A$. Since the CP is satisfiable, we
  have that $A\neq \emptyset$.  We will exhibit a sequence of events
  that, regardless of the initial configuration, leads to a satisfying
  assignment with a probability for which we find a lower bound.
    We consider the following sequence, divided in two phases:
 \begin{algorithmic}[1]
\STATE $t\leftarrow 0$
\REPEAT
\IF {$F_{\vec{x}({t})} \not\supseteq U_{\vec{x}({t})}$}
\STATE\textbf{Phase 1} Applying Lemma~\ref{lem:1}, after
$\tilde{t}\le N$ steps $F_{\vec{x}(t+ \tilde{t})} \supseteq
U_{\vec{x}(t+ \tilde{t})}$ and $F_{\vec{x}(t+\tilde{t})}\supsetneq
F_{\vec{x}(t)}$ (so $|F_{\vec{x}(t + \tilde{t})}| \ge |F_{\vec{x}(t)}|
+ 1$).  This event happens with probability greater than  $\gamma^{N^2} $.
\STATE $t\leftarrow t+\tilde{t}$
\ENDIF
\IF {$U_{\vec{x}(t)} \ne \emptyset$}
\STATE\textbf{Phase 2} We have $F_{\vec{x}(t)} \supseteq
U_{\vec{x}(t)} $.  Applying Lemma~\ref{lem:3}, after $\tilde{t}< N^2$ steps $|F_{\vec{x}(t + \tilde{t})}| =
|F_{\vec{x}(t)}| + 1$.  This event happens with probability greater than $\gamma^{N^3}$.  
\STATE $t\leftarrow t+\tilde{t}$
\ENDIF
\UNTIL{$U_{\vec{x}(t)} = \emptyset$ }.
 \end{algorithmic}
 This sequence is terminating, because the set $F_{\vec{x}(t)}$ is
 strictly increasing, and when $|F_{\vec{x}(t)}| = N$ we necessarily
 have $U_{\vec{x}(t)} = \emptyset$.  Each vertex will be added to
 $F_{\vec{x}(t)}$ only once, either by Phase 1 or Phase 2. 

 When a vertex is added by Phase 1, it will require at most $N$ steps
 and occur with probability at least $\gamma^{N(N+1)/2}$. When added by
 Phase 2, it will require at most $N^2$ steps and occur with
 probability at least $\gamma^{N^3}$. Since $N \le N^2$ and
 $\gamma^{N(N+1)/2} \ge \gamma^{N^3}$ for $N>1$, we can therefore upper bound the
 total number of steps by $N\cdot N^2 = N^3$ and lower bound the
 probability of the sequence by $(\gamma^{N^3})^N =
 \gamma^{N^4}$.

Due to the Markovian nature of  Algorithm~\ref{alg:gCFL}  and the
independence of the probability of the above sequence on its initial
conditions, if this sequence does not occur in $N^3$ iterations,
it has the same probability of occurring in the next $N^3$
iterations. The probability of convergence in $k\cdot N^3$ steps is greater than
$1-(1-\gamma^{N^4})^k$.   For $1-(1-\gamma^{N^4})^k \ge 1-\epsilon$ we require $k\le \frac{\log
  \epsilon}{\log(1-\gamma^{N^4})}\le - \frac{\log
  \epsilon}{\gamma^{N^4}} = e^{N^4 \log(\gamma^{-1})} \log(\epsilon^{-1})$
\end{proof}
 \begin{proof}[Proof of Corollary \ref{thm:coroll}]
   After running Phase 1 in the proof of Theorem \ref{thm:bounds} for the first time, we have
   $F_{\vec{x}(t+\tilde{t})}\supsetneq F_{\vec{x}(t)}$. If
   $U_{\vec{x}(t)} = \emptyset$ we have finished without running Phase
   2. Otherwise we must run Phase 2. But in this case we have from
   Lemma~\ref{lem:cycle} that $\C \neq \M$ (because there exists a
   pair of vertices $i,j$ such that $j\rightarrow k$ and $k \not
   \rightarrow j$), leading to a contradiction. So after Phase 1
   $U_{\vec{x}(t)} = \emptyset$ and Phase 2 is never executed.  The
   running time of Phase 1 is no greater than $N$ and occurs with
   probability at least $\gamma^{\sum_{k=1}^N k}= \gamma^{(N+1)N/2}$.

 \end{proof} 
\bibliographystyle{IEEEtranNAT}

\small\begin{thebibliography}{10}
\providecommand{\url}[1]{#1}
\csname url@rmstyle\endcsname
\providecommand{\newblock}{\relax}
\providecommand{\bibinfo}[2]{#2}
\providecommand\BIBentrySTDinterwordspacing{\spaceskip=0pt\relax}
\providecommand\BIBentryALTinterwordstretchfactor{4}
\providecommand\BIBentryALTinterwordspacing{\spaceskip=\fontdimen2\font plus
\BIBentryALTinterwordstretchfactor\fontdimen3\font minus
  \fontdimen4\font\relax}
\providecommand\BIBforeignlanguage[2]{{%
\expandafter\ifx\csname l@#1\endcsname\relax
\typeout{** WARNING: IEEEtran.bst: No hyphenation pattern has been}%
\typeout{** loaded for the language `#1'. Using the pattern for}%
\typeout{** the default language instead.}%
\else
\language=\csname l@#1\endcsname
\fi
#2}}

\bibitem[Duffy et~al.(2011)]{DBLP:journals/corr/abs-1103-3240}
K.~Duffy, C.~Bordenave, and D.~Leith, ``{{Decentralized Constraint
  Satisfaction}},'' \emph{CoRR}, vol. abs/1103.3240, 2011.

\bibitem[Barcelo et~al.(2011)]{jaume2011towards}
J.~Barcelo, B.~Bellalta, C.~Cano, A.~Sfairopoulou, M.~Oliver, and K.~Verma,
  ``{{Towards a Collision-free WLAN: Dynamic Parameter Adjustment in
  CSMA/E2CA}},'' \emph{{EURASIP Journal on Wireless Communications and
  Networking}}, 2011.

\bibitem[Fang et~al.(2010)]{fang2010decentralised}
M.~Fang, D.~Malone, K.~Duffy, and D.~Leith, ``{{Decentralised Learning MACs for
  Collision-free Access in WLANs}},'' \emph{{{Wireless Networks}}}, pp. 1--16,
  2010.

\bibitem[Checco et~al.(2012)]{checcoself}
A.~Checco, R.~Razavi, D.~Leith, and H.~Claussen, ``{{Self-Configuration of
  Scrambling Codes for WCDMA Small Cell Networks}},'' in \emph{IEEE 23rd
  International Symposium on Personal, Indoor and Mobile Radio Communications
  (PIMRC)}, Sydney, Australia, September 2012.

\bibitem[Raniwala and Chiueh(2005)]{raniwala2005architecture}
A.~Raniwala and T.~Chiueh, ``{{Architecture and Algorithms for an IEEE
  802.11-based Multi-channel Wireless Mesh Network}},'' in \emph{{{INFOCOM
  2005. 24th Annual Joint Conference of the IEEE Computer and Communications
  Societies. Proceedings IEEE}}}, vol.~3.\hskip 1em plus 0.5em minus
  0.4em\relax IEEE, 2005, pp. 2223--2234.

\bibitem[Mishra et~al.(2006)]{mishra2006distributed}
A.~Mishra, V.~Shrivastava, D.~Agrawal, S.~Banerjee, and S.~Ganguly,
  ``{{Distributed Channel Management in Uncoordinated Wireless
  Environments}},'' in \emph{{{Proceedings of the 12th Annual International
  Conference on Mobile Computing and Networking}}}.\hskip 1em plus 0.5em minus
  0.4em\relax ACM, 2006, pp. 170--181.

\bibitem[Mishra et~al.(2006)]{mishra2006client}
A.~Mishra, V.~Brik, S.~Banerjee, A.~Srinivasan, and W.~Arbaugh, ``{{A
  Client-driven Approach for Channel Management in Wireless LANs}},'' in
  \emph{{{in INFOCOM 2007. 25th Conference on Computer Communications}}},
  vol.~6, 2006.

\bibitem[Leung and Kim(2003)]{leung2003frequency}
K.~Leung and B.~Kim, ``{{Frequency Assignment for IEEE 802.11 Wireless
  Networks}},'' in \emph{{{Vehicular Technology Conference, 2003. VTC
  2003-Fall. 2003 IEEE 58th}}}, vol.~3.\hskip 1em plus 0.5em minus 0.4em\relax
  IEEE, 2003, pp. 1422--1426.

\bibitem[Narayanan(2002)]{narayanan2002}
L.~Narayanan, \emph{{{Handbook of Wireless Network and Mobile
  Computing}}}.\hskip 1em plus 0.5em minus 0.4em\relax {{Wiley Series on
  Parallel and Distributed Computing}}, 2002, ch. {{Channel Assignment and
  Graph Multicoloring}}.

\bibitem[Dousse(2012)]{dousse2012percolation}
O.~Dousse, ``{{Percolation in Directed Random Geometric Graphs}},'' in
  \emph{{{Information Theory Proceedings (ISIT), 2012 IEEE International
  Symposium on}}}.\hskip 1em plus 0.5em minus 0.4em\relax IEEE, 2012, pp.
  601--605.

\bibitem[Kothapalli et~al.(2006)]{kothapalli2006distributed}
K.~Kothapalli, M.~Onus, C.~Scheideler, and C.~Schindelhauer, ``{{Distributed
  Coloring in $\mathcal{O}(\sqrt{\log N})$-bits}},'' in \emph{{{Proc. of IEEE
  International Parallel and Distributed Processing Symposium (IPDPS)}}}, 2006.

\bibitem[Hedetniemi et~al.(2002)]{hedetniemi2002fault}
S.~Hedetniemi, D.~Jacobs, and P.~Srimani, ``{{Fault Tolerant Distributed
  Coloring Algorithms that Stabilize in Linear Time}},'' in \emph{{{Proceedings
  of the IPDPS-2002 Workshop on Advances in Parallel and Distributed
  Computational Models}}}, 2002, pp. 1--5.

\bibitem[Johansson(1999)]{johansson1999simple}
{\"O}.~Johansson, ``{{Simple Distributed $\Delta +1$-coloring of Graphs}},''
  \emph{{{Information Processing Letters}}}, vol.~70, no.~5, pp. 229--232,
  1999.

\bibitem[Kauffmann et~al.(2007)]{kauffmann2007measurement}
B.~Kauffmann, F.~Baccelli, A.~Chaintreau, V.~Mhatre, K.~Papagiannaki, and
  C.~Diot, ``{{Measurement-based Self Organization of Interfering 802.11
  Wireless Access Networks}},'' in \emph{{{INFOCOM 2007. 26th IEEE
  International Conference on Computer Communications. IEEE}}}.\hskip 1em plus
  0.5em minus 0.4em\relax IEEE, 2007, pp. 1451--1459.

\bibitem[Kauffmann et~al.(2007)]{kauffmann2007self}
------, ``{{Self Organization of Interfering 802.11 Wireless Access
  Networks}},'' in \emph{{{INFOCOM 2007. 26th IEEE International Conference on
  Computer Communications. IEEE}}}, 2007.

\bibitem[Clifford and Leith(2007)]{clifford2007channel}
P.~Clifford and D.~Leith, ``{{Channel Dependent Interference and Decentralized
  Colouring}},'' \emph{{{Network Control and Optimization}}}, pp. 95--104,
  2007.

\bibitem[Leith et~al.(2012)]{leith2012wlan}
D.~Leith, P.~Clifford, V.~Badarla, and D.~Malone, ``{{WLAN Channel Selection
  Without Communication}},'' \emph{{{Computer Networks}}}, 2012.

\bibitem[Grimmett(1999)]{grimmett1999percolation}
G.~Grimmett, \emph{{{Percolation}}}.\hskip 1em plus 0.5em minus 0.4em\relax
  Springer, 1999, vol. 321.

\bibitem[ue2({2004})]{ue2004}
``{{UE Radio transmission and Reception (FDD)}},'' {{3GPP TS25.101}}, Tech.
  Rep., {2004}.

\bibitem[wig(2010)]{wigle}
\BIBentryALTinterwordspacing
``wigle.net,'' 2010. [Online]. Available: \url{http://www.wigle.net/}
\BIBentrySTDinterwordspacing

\bibitem[802(1997)]{80211}
\emph{{{IEEE Standard for Wireless LAN Medium Access Control (MAC) and Physical
  Layer (PHY) Specifications , Nov. 1997. P802.11}}}, Std., 1997.

\bibitem[Goldsmith(2005)]{goldsmith2005wireless}
A.~Goldsmith, \emph{{{Wireless Communications}}}.\hskip 1em plus 0.5em minus
  0.4em\relax {{Cambridge University Press}}, 2005.

\end{thebibliography}


\end{document}